\documentclass[12pt]{article}
\usepackage{times}
\usepackage[utf8]{inputenc}
\usepackage[english]{babel}

\usepackage[margin=1in]{geometry}
\usepackage{graphicx}
\usepackage[pdftex,colorlinks=true,linkcolor=blue,citecolor=blue,urlcolor=black]{hyperref}
\usepackage{amsmath, amsthm, amssymb}
\usepackage{subcaption}
\usepackage{comment}
\usepackage{url}
\usepackage[ruled,lined,linesnumbered,titlenumbered]{algorithm2e}
\usepackage{tikz}
\usepackage{ragged2e}

\title{Quantum walk search algorithms and effective resistance}
\author{Stephen Piddock\footnote{School of Mathematics, University of Bristol, UK and Heilbronn Institute for Mathematical Research, Bristol; \texttt{stephen.piddock@bristol.ac.uk}}}

\newcommand{\ket}[1]{| #1 \rangle}
\newcommand{\bra}[1]{\langle #1|}
\newcommand{\braket}[1]{\langle #1 \rangle}

\newcommand{\proj}[1]{| #1 \rangle \langle #1 |}

\DeclareMathOperator{\poly}{poly}

\DeclareMathOperator{\linspan}{span}

\newcommand{\norm}[1]{\left\lVert#1\right\rVert}

\newcommand{\be}{\begin{equation}}
\newcommand{\ee}{\end{equation}}
\newcommand{\bea}{\begin{eqnarray}}
\newcommand{\eea}{\end{eqnarray}}
\newcommand{\bes}{\begin{equation*}}
\newcommand{\ees}{\end{equation*}}
\newcommand{\beas}{\begin{eqnarray*}}
\newcommand{\eeas}{\end{eqnarray*}}

\makeatletter
\newtheorem*{rep@theorem}{\rep@title}
\newcommand{\newreptheorem}[2]{%
\newenvironment{rep#1}[1]{%
 \def\rep@title{#2 \ref{##1} (restated)}%
 \begin{rep@theorem}}%
 {\end{rep@theorem}}}
\makeatother

\newtheorem{thm}{Theorem}
\newtheorem*{thm*}{Theorem}

\newtheorem{lem}[thm]{Lemma}
\newtheorem*{lem*}{Lemma}

\newreptheorem{thm}{Theorem}
\newreptheorem{lem}{Lemma}
\newreptheorem{cor}{Corollary}

\begin{document}

\maketitle
           
\begin{abstract}
We consider the problem of finding a marked vertex in a graph from an arbitrary starting distribution, using a quantum walk based algorithm. 
We work in the framework introduced by Belovs which showed how to detect the existence of a marked vertex in $O(\sqrt{RW})$ quantum walk steps, where $R$ is the effective resistance and $W$ is the total weight of the graph. 
Our algorithm outputs a marked vertex in the same runtime up to a logarithmic factor in the number of marked vertices.
When starting in the stationary distribution, this recovers the recent results of Ambainis et al \cite{Ambainis2019}.
We also describe a new algorithm to estimate the effective resistance $R$.
\end{abstract}       
 
\section{Introduction}

In this paper we study how quantum walks can be used to find a marked vertex in a graph, starting from an arbitrary initial distribution.
This problem has been well studied for the special case of starting in the stationary distribution, where it is known that a marked vertex can be found in time $O(\sqrt{HT})$ quantum walk steps, where $HT$ is the classical hitting time - the expected number of steps taken by a classical random walk to find a marked vertex.

Szegedy \cite{Szegedy2004} originally showed that a quantum walk could detect the existence of a marked element in $O(\sqrt{HT})$ quantum walk steps starting from the stationary distribution. 
However the algorithm would only output yes/no, and in particular in a yes instance, it would not output an example of a marked vertex. 
The problem of actually \emph{finding} a marked vertex has been studied for a number of examples, and eventually it was shown to be possible in $O(\sqrt{HT})$ steps for an arbitrary graph with a single marked vertex by  Krovi, Magniez, Ozols and Roland \cite{Krovi2016}, and later by Dohotaru and H{\o}yer \cite{Dohotaru2017}.
Only recently was this result extended by Ambainis, Gily{\'e}n, Jeffery and Kokainis \cite{Ambainis2019} to the case of multiple marked vertices.

Belovs introduced a quantum walk algorithm \cite{Belovs2013} which can start from an arbitrary initial distribution $\sigma$ on a graph and decides if there is a marked vertex in the graph. This algorithm differs from the other quantum walk algorithms discussed so far which first require the state of the system to be prepared in the stationary state of the graph (which might itself take a long time).

Belovs' algorithm runs in time $O(\sqrt{R_{\sigma,M}W})$, where $R_{\sigma,M}$ is the effective resistance between $\sigma$ and the set of marked vertices $M$, and $W$ is the total weight of the graph $G$.
But as with Szegedy's original algorithm, Belovs' algorithm only outputs whether or not there is a marked vertex; it does not output an example of a marked vertex when one exists.

In this paper we present an algorithm based on the framework of Belovs \cite{Belovs2013}, which outputs a marked vertex in approximately the same time, up to logarithmic factors in the number of marked elements.

\begin{thm}
\label{thm:findm}
Let $G$ be a weighted graph with a set of marked vertices $M$. 
Given an upper bound $\tilde{W}$ on the total weight of $G$, and starting in an initial distribution $\sigma$, there is a quantum algorithm that outputs a marked vertex in $O\left(\sqrt{R_{\sigma,M}\tilde{W}}\log^3(|M|)\right)$ quantum walk steps.
\end{thm}
By considering the special case where $\sigma$ is the stationary distribution, we recover the results of Ambainis, Gily{\'e}n, Jeffery and Kokainis \cite{Ambainis2019}, up to the $\log(|M|)$ factors. 
However our proof is very different and (arguably) simpler, although it is difficult to compare the two since we are working in a different framework.

Theorem~\ref{thm:findm} requires an upper bound $\tilde{W}$ on the total weight $W$ of $G$. When no such bound is known, it is still possible to use the algorithm of Theorem~\ref{thm:findm} with successively larger values of guesses for $W$ (doubling each time), at a cost of an extra logarithmic factor in the runtime. 
See Appendix~\ref{sec:unknownW} for a more detailed discussion of this.


The algorithm works in much the same way as Belovs' algorithm. The idea is to define a quantum walk operator on a closely related graph $G'$ and perform phase estimation. When the phase estimation ancilla register is measured and outputs the all zero string, Belovs' algorithm declares a marked element has been found. 
Here, one of our main technical contributions is to show that when this happens, the remaining register is left (approximately) in a state $\ket{\Phi}$ corresponding to the electric flow through the graph. 
Then by making suitable further adjustments to $G'$, we can show that simply measuring $\ket{\Phi}$ in the standard basis will give a marked vertex with high probability.

A subroutine used in our algorithm estimates the effective resistance $R$ to constant multiplicative accuracy. In Section~\ref{sec:estimateR} we describe how to improve this method to obtain a higher accuracy approximation:
\begin{thm}
\label{thm:estimateR}
Let $G$ be a weighted graph with a set of marked vertices $M$. 
Given an upper bound $\tilde{W}$ on the total weight of $G$,and a starting vertex $s$, there is a quantum algorithm that estimates the effective resistance $R_{s,M}$ between $s$ and $M$ to multiplicative error $\epsilon$ in  $O(\sqrt{R_{s,M}\tilde{W}}/\epsilon^2)$ quantum walk steps.
\end{thm}

We hope that the $\epsilon$ dependence of this algorithm can be improved to give a $O(\sqrt{R\tilde{W}}\frac{1}{\epsilon}\log(\frac{1}{\epsilon}))$ time algorithm, by performing the phase estimation to lower accuracy but repeated many times - see Section~\ref{sec:estimateR} for a brief discussion.

This could be compared to the $\tilde{O}(n\sqrt{R_{s,t}}/\epsilon^\frac{3}{2})$ runtime from \cite{Ito2019} by Ito and Jeffery using approximate span programs, for the case where $M=\{t\}$ contains a single vertex. However it is worth noting that they are working in the adjacency query model, whereas we are counting the number of uses of the quantum walk operator, which is more closely related to the edge list model - although even then there is still a dependence on the maximum degree of the graph for actually implementing the quantum walk operator.

Belovs' algorithm was used by Montanaro \cite{Montanaro2015} to give a quantum speed-up for classical backtracking techniques. Here, the graph to be searched is the computational tree defined by the classical backtracking algorithm. For this application, it is crucial that the algorithm starts at the root of the tree, and not the stationary distribution over the whole graph.

For the special case of starting at a particular vertex in a tree with a single marked vertex, Montanaro showed that phase estimation can produce the electrical flow state $\ket{\Phi}$. This was extended to trees with multiple marked vertices by Jarret and Wan \cite{Jarret2017}.
The proof of this fact contained in \cite{Jarret2017} relied on certain special properties of electrical flows that only hold for trees. Here we are able to prove this cleanly and directly for all graphs by observing a connection to the electrical potential or voltages in the graph.

We note that our algorithm may not be the only useful thing to do with the electrical flow state $\ket{\Phi}$. 
For instance, one could measure the flow state in the computational basis and use the outcome to choose a new starting vertex.
This is essentially the idea of the algorithm used in \cite{Montanaro2015,Jarret2017}, but the analysis of the classical random process for generic graphs appear to be much more complicated than for the case of trees, and we leave this for future work.

The state $\ket{\Phi}$ contains a lot of information about the paths from $s$ to $M$ including which edges are ``more important'' than others.
Perhaps there are more intelligent quantum algorithms that could be applied to this state in order to extract more of this information in a useful way.

\subsection{Related work}
As this paper was being prepared, we became aware of concurrent and independent work by Gily{\'e}n, Jeffery and Apers \cite{Gilyen2019}. They obtain a result similar to our Theorem~\ref{thm:findm}, although the proof uses quite different ideas, building on the methods of \cite{Ambainis2019}.

\section{Preliminaries}

\subsection{Graphs and electrical networks}
Let $G=(E,V)$ be  a weighted undirected graph with edge weights $\{w_e \}_{e \in E}$.
Let $W$ be the total weight $W= \sum_{e \in E} w_e$. 
Let $\sigma=\{\sigma_x\}_{v \in V}$ be the initial probability distribution over vertices.

We will assume that the graph is bipartite with parts $A$ and $B$, and we assume that $\sigma_x=0$ for $x \in B$.
As noted in~\cite{Belovs2013}, this is not a very restrictive assumption: if the graph $G$ is not bipartite, then consider a new graph $G'$ with vertex set $V \times \{0,1\}$ and edges between $(x,0)$ and $(y,1)$ if and only if there is an edge between $x$ and $y$ in G. 
This increases the total number of vertices, the total weight, and the effective resistance, each by a factor of at most 2.

We can view the graph as an electrical network, where each edge $e$ represents a wire with a resistor of resistance $1/ w_e$.
This is an idea that has been very fruitful in graph theory, in particular in connection with the hitting time of classical random walks, see~\cite{Bollobas1998}(Chapters 2 and 9) for more details.
Then we say a \emph{unit flow} from $\sigma$ to $M$ is any assignment $\{p_{xy}\}_{xy \in E} $ of real numbers to the edges of the graph such that:
\[\sum_{y: (xy) \in E} p_{xy}= \sigma_x \qquad \forall x \notin M\]
For the special case where $\sigma$ is concentrated at a single vertex $s$, this ensures that the flow out of $s$ is 1, and that the flow is conserved at all other unmarked vertices.
\footnote{Strictly speaking one must assign a direction to each edge in the underlying graph for this equation to make sense. Since the graph is bipartite, we choose the direction of each edge to go from $A$ to $B$.}

The energy of a flow is given by 
\begin{equation}
\sum_{xy \in E} p_{xy}^2/w_{xy}.
\label{eq:energy}
\end{equation}
The \emph{effective resistance} $R_{\sigma,M}$ is given by the minimal energy over all unit flows from $s$ to $M$. 
The \emph{electrical flow} from $\sigma$ to $M$ is the unit flow $\{f_{xy}\}_{xy}$ from $\sigma$ to $M$ which achieves this minimal energy $R_{\sigma,M}$.

Another property of the electrical flow, which will be useful for us, is that it is also the unique flow for which there exists $\{v_x\}_{x \in V}$ such that $f_{xy}=(v_x-v_y)w_{xy}$ and $v_m=0$ for all $m \in M$.
The $\{v_x\}_{x \in V}$ is called the \emph{voltage} or \emph{potential}.
For the special case where $\sigma$ is concentrated at a single vertex $s$, we have that $v_s=R_{s,M}$.
For more information on graph theory and electrical flows, see \cite{Bollobas1998}.

One reason considering graphs as electrical networks is the connection to classical random walks. A classical random walk on a weighted graph $G$ is the Markov chain defined by the probability transfer  matrix
\[P_{x\rightarrow y}=\frac{w_{xy}}{d_x}\]
where $d_x= \sum_{y:xy \in E} w_{xy}$ is the weighted degree of $x$. 
This means when at a vertex $x$ the probability of moving to vertex $y$ is $\frac{w_{xy}}{d_x}$.
Starting in a distibution $\sigma$, let $HT_{\sigma, M}$ be the hitting time, the expected number of steps unitl you are at a vertex in $M$. Then the following relations between hitting times and effective resistance are known:
\begin{enumerate}
\item When $\sigma $ is concentrated at a single vertex $s$ and $M=\{t\}$, 
\[HT_{s,t}+HT_{t,s}=2R_{s,t}W\]
\item When $\sigma=\pi$ the stationary distribution, 
\[HT_{\pi,M}=2R_{\pi,M}W\]
\end{enumerate}
This second fact shows how Theorem~\ref{thm:findm} recovers the quadratic speed-up (up to the $\log|M|$ factors) recently proven by Ambainis et al. \cite{Ambainis2019} for the case of starting in the stationary distribution.
\subsection{Quantum walk operator}
\label{sec:QWdef}
Next we define the quantum walk operator for a weighted graph $G$, with a set of marked vertices $M$ and a special starting vertex $s \in B$.
For our algorithm we will explain in Section~\ref{sec:G'} how to take a graph $G$ with a set of marked vertices $M$ and initial distribution $\sigma$ and construct a graph $G'$ with a set of marked vertices $M'$ and special starting vertex $s' \in B$ for which we can define a quantum walk operator in this way.
This construction then matches that presented by Belovs in \cite{Belovs2013}.

The quantum walk operator acts on a Hilbert space $\mathcal{H}$ with a basis state for each edge in the graph:
\[\mathcal{H}= \linspan \left( \{\ket{xy}\}_{xy \in E} \right)\]

For each vertex $x \in V$, define an unnormalised state $\ket{\phi_x}$, a weighted sum of all edges incoming to $x$:
\[\ket{\phi_x}=\sum_{y:xy \in E} \sqrt{w_{xy}} \ket{xy}.\]

For each vertex $x \in V$, we define the diffusion operator $D_x$ as follows:
if $x$ is a marked vertex or  the starting vertex $s$, then $D_x$ is the identity;
otherwise $D_x$ is the reflection about $\ket{\phi_x}$. 
More precisely: 
\begin{equation*}
D_x= \begin{cases}
I & x \in M \cup \{s\}\\
I- 2\proj{\psi_x} & x \notin M \cup \{s\}  
\end{cases} 
\qquad \text{ where } \ket{\psi_x}= \frac{1}{\sqrt{d_x}}  \ket{\phi_x} 
\end{equation*}
where $d_x= \sum_{y:xy \in E} w_{xy}$ is the weighted degree of the vertex $x$.

Then we define unitary reflection operators $U_A=\bigoplus_{x\in A}D_x$ and $U_B= \bigoplus_{x\in B}D_x$. 
The quantum walk operator is the product of these two reflections $U_AU_B$.
In previous work, such as \cite{Belovs2013}, these operators are usually denoted $R_A$ and $R_B$.
We have chosen to switch to uppercase $U$ notation to avoid confusion with the uppercase $R_{\sigma,M}$ notation which we are using for effective resistance in this paper.

\subsection{Tools}

We will make use of the following standard results:

\begin{lem}[Effective spectral gap lemma \cite{Lee2010}]
\label{lem:effspecgap}
Let $\Pi_A$ and $\Pi_B$ be projectors and let $U_A=2\Pi_A -I$, $U_B=2\Pi_B-I$. Let $P_{\epsilon}$ be the projector onto the span of the eigenvectors of $U_A U_B$ with eigenvalues $e^{2i \theta}$ such that $|\theta| \le \epsilon$. Then for any $\epsilon>0$, and any vector $\ket{\psi}$ such that $\Pi_A \ket{\psi}=0$, we have 
\[ \norm{P_{\epsilon} \Pi_B \ket{\psi}} \le \epsilon \norm{\ket{\psi}} \|.\]
\end{lem}

\begin{lem}[Phase estimation, \cite{Cleve1998,Kitaev1995}]
	\label{lem:PE}
	For every integer $t$ and every unitary $U$ on $n$ qubits, 
	there exists a uniformly generated circuit $V$ such that $V$ acts on $n+t$ qubits and:
	\begin{enumerate}
		\item $V$ uses the controlled-$U$ operation $O(2^t)$ times and contains $O(t^2)$ other gates.
		\item For every eigenvector of $\ket{\psi}$ of $U$ with eigenvalue 1, $V\ket{\psi}\ket{0^t}=V\ket{\psi}\ket{0^t}$.
		\item If $U\ket{\psi}=e^{2i \theta}\ket{\psi}$ where $\theta \in (0,\pi)$ then $V\ket{\psi}\ket{0^t}=\ket{\psi}\ket{w}$ where $\ket{w}$ satisfies $|\braket{w|0^t}|^2=O(1/2^t \theta)$.
	\end{enumerate}
	
\end{lem}
\section{Electrical flow state}
\label{sec:electricflow}
Let $f_{xy}$ denote the flow from $x$ to $y$ in the electrical flow of unit current from $s$ to $M$. 
We now define the electrical flow state $\ket{\Phi}$ as 
\begin{equation}
\ket{\Phi}=\frac{1}{\sqrt{R_{s,M}}}\sum_{xy} f_{xy}/ \sqrt{w_{xy}} \ket{xy} 
\end{equation} 
This state is normalised since by the definition of effective resistance $R_{s,M}$ 
\[\|\ket{\Phi}\|^2=\frac{\sum_{xy} f_{xy}^2/w_{xy}}{R_{s,M}}=1.\]

Next we show that projecting $\ket{\psi_s}$ onto the small eigenvalues of the quantum walk operator $U_AU_B$ results in the state $\ket{\Phi}$ up to some small error. Therefore running phase estimation to good accuracy on $U_AU_B$ when starting in the state $\ket{\psi_s}$ and getting the all zero outcome will result (approximately)  with the state $\ket{\Phi}$.
\begin{lem}
\label{lem:Phi+-}
Let $P_{\epsilon}$ be the projector onto the span of the eigenvectors of $U_A U_B$ with eigenvalues $e^{2i \theta}$ such that $|\theta| \le \epsilon$. Then the following is true:
\begin{enumerate}
\item The state $\ket{\Phi}$ is an eigenvector of $U_A U_B$ with eigenvalue 1, $P_0 \ket{\Phi}=\ket{\Phi}$.
\item For any $\epsilon \ge 0$, 
\[\norm{P_{\epsilon} \ket{\psi_s}+\frac{1}{\sqrt{R_{s,M}d_s}}\ket{\Phi}} \le \epsilon\sqrt{\frac{\sum_{x\in A}v_x^2d_x}{R_{s,M}^2d_s}}
.\]
\end{enumerate}
where $\{v_x\}_{x \in V}$ is the voltage of the electric flow from $s$ to $M$.
\end{lem}

To prove the second part of Lemma~\ref{lem:Phi+-}, we will need the following technical Lemma (from \cite{Lee2010}):

\begin{proof}[Proof (of Lemma~\ref{lem:Phi+-})]
First we prove the first claim. Let $x$ be any vertex in the graph, then
\[\braket{\phi_x | \Phi}=\left(\sum_{y: (xy) \in E}  \sqrt{w_{xy}}\bra{xy}\right) \left( \sum_{(uv) \in E} \frac{f_{uv}}{\sqrt{R_{s,M}w_{uv}}} \ket{uv}\right)=\frac{1}{\sqrt{R_{s,M}}}\sum_{y: (xy) \in E} f_{xy}\]
When  $x \notin M\cup \{s\}$, then the flow is conserved at $x$ and so $\braket{\psi_x|\Phi}$ is proportional to 
$\sum_{y: (xy) \in E} f_{xy}=0$.
Therefore $\ket{\Phi}$ satisfies $U_A \ket{\Phi}=\ket{\Phi}$ and $U_B \ket{\Phi}=\ket{\Phi}$, and so $U_AU_B \ket{\Phi}=\ket{\Phi}$ as claimed.

Now we prove the second claim.
By part 1, $P_{\epsilon}\ket{\Phi}=\ket{\Phi}$ for any $\epsilon\ge0$, so 
\[P_{\epsilon} \ket{\psi_s}+\frac{1}{\sqrt{R_{s,M}d_s}}\ket{\Phi} =P_{\epsilon}\left( \ket{\psi_s}+\frac{1}{\sqrt{R_{s,M}d_s}}\ket{\Phi}\right).\]
In order to apply Lemma~\ref{lem:effspecgap}, we need to find a vector $\ket{\psi}$ such that $\Pi_A \ket{\psi}=0$ and $\Pi_B\ket{\psi}$ is proportional to $\ket{\psi_s}+\frac{1}{\sqrt{R_{s,M}d_s}}\ket{\Phi}$.
Since $f_{xy}$ is an   electrical flow there is an assignment $\{v_x\}_{x \in V}$ to the vertices of $G$, known as the voltage or potential difference such that $f_{xy}=(v_x-v_y)w_{xy}$ for all $xy \in E$; $v_s=R_{s,M}$, and  $v_m=0$ for any $m \in M$. 
See \cite{Bollobas1998} for more details on the theory of electrical networks.

We use this potential $v$ to define an (unnormalised) vector $\ket{\psi}$ to which we can apply Lemma~\ref{lem:effspecgap}.
Let $\ket{\psi}=\sum_{x \in A} v_x \ket{\phi_x}$
which clearly satisfies $\Pi_A \ket{\psi}=0$.

Before calculating $\Pi_B \ket{\psi}$, we first note that for $x \in A$ and $y \in B$,  $\braket{\phi_y|\phi_x} =w_{xy}$. Then for any $y \in B\backslash (M\cup\{s\})$, we have

\begin{align*}
(\proj{\psi_y})\ket{\psi}
&=\frac{1}{d_y}\proj{\phi_y}\left(\sum_{x \in A} v_x \ket{\phi_x}\right)\\
&= \frac{1}{d_y} \ket{\phi_y}\sum_{x \in A} v_x w_{xy}= \frac{1}{d_y} \ket{\phi_y}\sum_{x \in A} v_y w_{xy}\\
&= v_y \ket{\phi_y}
\end{align*}
where the equality in the second line holds because $y \in B\backslash (M\cup\{s\})$ and so the flow is conserved at $y$ implying that $\sum_{x \in A} (v_x-v_y) w_{xy}=0$. The final equality follows from the definition of the weighted degree $d_y=\sum_{x \in A} w_{xy}$.

Now we calculate $\Pi_B \ket{\psi}$:
\begin{align*}
\Pi_B \ket{\psi}&=\left(I-\sum_{y \in B\backslash (M\cup\{s\})} \proj{\psi_y}\right) \ket{\psi}\\
&= \sum_{x \in A} v_x \ket{\phi_x}-\sum_{y \in B\backslash (M\cup\{s\})} v_y \ket{\phi_y}\\
&=v_s \ket{\phi_s}+\sum_{y \in B \cap M} v_y \ket{\phi_y} + \sum_{x \in A, y \in B} (v_x-v_y) \sqrt{w_{xy}}\ket{xy}\\
&=R_{s,M}\sqrt{d_s}\ket{\psi_s} +0+ \sqrt{R_{s,M}}\ket{\Phi}
\end{align*}
where the middle term is zero since $v_m=0$ for any vertex $m \in M$.

%

Therefore, applying Lemma~\ref{lem:effspecgap} to $\frac{1}{R_{s,M}\sqrt{d_s}}\ket{\psi}$ gives 
\[ \norm{P_{\epsilon}\left( \ket{\psi_s}+\frac{1}{\sqrt{R_{s,M}d_s}}\ket{\Phi}\right)}  \le \epsilon \frac{1}{R_{s,M}\sqrt{d_s}} \norm{\ket{\psi}}
= \epsilon \sqrt{\frac{\sum_{x \in A} v_x^2 d_x}{R_{s,M}^2d_s}}\]
 where the final equality follows from the definition of $\ket{\psi}$.

\end{proof}

\section{Algorithm for finding marked vertices}

\subsection{Augmented graph $G'$}
\label{sec:G'}
The algorithm runs by performing phase estimation on the quantum walk operator corresponding to an altered graph $G'$ defined here.

%
%
%

Given a graph $G$, starting distribution $\sigma$, and a set of marked vertices $M$, we make two alterations to define an alternative graph $G'$ depending on two parameters $x$ and $\eta$ as follows:
\begin{enumerate}
\item We add an additional vertex $s'$ and edges of weight $w_{s'u}=\sqrt{\sigma_u}/\eta$. This is the addition made by Belovs to define the quantum walk operator in \cite{Belovs2013}.
\item For each vertex $k \in M$, we add an additional vertex $k'$ and an edge of weight $w_{kk'}=1/x$. Let $M'$ be the set of additional vertices of this form.
This addition is new in this work.
\end{enumerate}

We now define the quantum walk operator as in Section~\ref{sec:QWdef} with respect to this new graph $G'$, where $s'$ is the special starting vertex and $M'$ is the marked vertex set.

We assume we have access to an initial state $\ket{\psi_{s'}}=\sum_{u} \sqrt{\sigma_u} \ket{s'u}$ corresponding to the classical probability distribution $\sigma$.
The challenge is to choose sensible choices of $\eta$ and $x$ so that performing phase estimation and measuring in the computational basis will, with high probability, return an edge $(kk')$ for some $k \in M$.

\begin{lem}
\label{lem:PEG'}
Let $U_A$ and $U_B$ be the quantum walk operators for $G'$, with $\{s'\} \cup M'$ the set of vertices where $D_x$ acts trivially.
Then running phase estimation on $U_AU_B$ starting in the state $\ket{\psi_{s'}}=\sum_{u} \sqrt{\sigma_u} \ket{s'u}$ for time $\gamma \sqrt{\eta W+1}/\epsilon$ (for some constant $\gamma$ to be determined) with $t=\log(\gamma \sqrt{\eta W+1}/\epsilon)$ ancilla bits, outputs $0^t$ with probability $p'$
\[\frac{\eta}{R_{s',M'}}\le p' \le \frac{\eta}{R_{s',M'}}+\epsilon\]
leaving the other register in a state $\ket{\tilde{\Phi}}$ such that 
\[\frac{1}{2}\norm{ \proj{\tilde{\Phi}}-\proj{\Phi}}_1 \le \sqrt{\frac{\epsilon}{p'}}\]
\end{lem} 

\begin{proof}
The proof follows immediately from Lemma~\ref{lem:PEnew} (in Appendix~\ref{sec:PE}) once we have checked that 
\begin{equation}\norm{P_{\epsilon} \ket{\psi_s'}+\sqrt{\frac{\eta}{R_{s',M'}}}\ket{\Phi}} \le \epsilon O(\sqrt{\eta W+1})  \qquad \forall \epsilon\ge 0
\end{equation}
By Lemma~\ref{lem:Phi+-}, we know that for any $\epsilon \ge 0$, 
\[\norm{P_{\epsilon} \ket{\psi_s'}+\frac{1}{\sqrt{R_{s',M'}d_{s'}}}\ket{\Phi}} \le \epsilon\sqrt{\frac{\sum_{x\in A}v_x^2d_x}{R_{s',M'}^2d_{s'}}}
.\]
Observing that $d_{s'}=\sum_{u} \sigma_u/\eta =1/\eta$, this is equivalent to 
\[\norm{P_{\epsilon} \ket{\psi_s'}+\sqrt{\frac{\eta}{R_{s',M'}}}\ket{\Phi}} \le \epsilon\sqrt{\eta\frac{\sum_{x\in A}v_x^2d_x}{R_{s',M'}^2}}
.\]
To complete the proof, we need to bound $\sum_{a\in A}v_a^2d_a$
\begin{align}
\sum_{a\in A}v_a^2d_a&=\sum_{ab \in E(G')} v_a^2 w_{ab}\\
&=\sum_{u} v_u^2 w_{s'u}+\sum_{ab \in E(G)} v_a^2 w_{ab}+\sum_{k \in M\cap A} v_{k}^2w_{kk'}\\
&=\sum_{u} v_u^2 \sigma_u/\eta+\sum_{ab \in E(G)} v_a^2 w_{ab}+\sum_{k \in M\cap A} f_{kk'}^2/w_{kk'}
\end{align}
For the first two terms, we use the fact that $v_a \le R_{s',M'}$ for all $a$, and thus these terms are bounded by $R_{s',M'}^2(1/\eta+W)$. The third term is bounded by $R_{s',M'}$ by observing that this term is just part of the positive sum Eq(\ref{eq:energy}) that makes up the definition of $R_{s',M'}$ .
\[\norm{P_{\epsilon} \ket{\psi_s'}+\sqrt{\frac{\eta}{R_{s',M'}}}\ket{\Phi}} \le \epsilon\sqrt{\eta\frac{\sum_{x\in A}v_x^2d_x}{R_{s',M'}^2}}
\le \epsilon \sqrt{\eta\left(\frac{1}{\eta} +W +\frac{1}{R_{s',M'}}\right)}\le \epsilon \sqrt{\eta W+2}\]
\end{proof}

\subsection{Algorithm to find $\eta$ such that $\eta/R\approx 1/2$}
\label{sec:findeta}
In this section, we only ever deal with the case when $x=0$, where effectively we have not made adjustment 2. in the definition of $G'$, and we can identify $M$ and $M'$. As previously discussed, given the analysis of Section~\ref{sec:electricflow}, we wish to run phase estimation on $U_AU_B$ and the state $\ket{\psi_s'}$ and get the outcome $\ket{0^t}$ in order to approximately prepare the electric flow state $\ket{\Phi}$ which we will then measure.
We need to make a sensible choice of $\eta$, such that the probability of getting $\ket{0^t}$ is $\Omega(1)$.

Here we present an algorithm for finding such an $\eta$.

\begin{algorithm}
\caption{Find $\eta$ such that $\eta/R_{s',M'}\approx 1/2$}
\label{alg:findeta}
\KwIn{Graph $G$, initial distribution $\sigma$ over vertices, upper bound $W$ on total weight,}
Set $\eta=1/W$ and $x=0$\;
Construct $G'$ and $U_AU_B$ for this choice of $\eta,x$\;
Starting in $\ket{\psi_s'}$, run phase estimation on $U_AU_B$ for time $\sqrt{\eta W}$ \;
Do amplitude estimation on the all zero string on the phase estimation register to constant accuracy to get an estimate $\tilde{a}$ of $\eta/R_{s',M'}$.\;
If $\tilde{a} \le 1/2$, double $\eta$ and return to Step 2. Otherwise, output $\eta$\;
\end{algorithm}
 The correctness of this algorithm follows from Lemma~\ref{lem:PEG'}.
 To ensure a probability of success at least $1-\delta$, we can do the standard trick of repeating the amplitude estimation step $O(\log 1/\delta)$ times and taking the median.
A more careful analysis of this algorithm is is presented by Jarret and Wan in \cite{Jarret2017}.

We note that for any $\eta$, we have $R_{s',M}\le \eta + R_{\sigma,M}$ because $\eta + R_{\sigma,M}$ is the energy of the unit flow from $s'$ to $M$ of the form: $\sigma_u$ flow is sent from $s'$ to $u$, and then the electric flow from $\sigma$ to $M$. $R_{s',M}$ is the minimal energy of any flow from $s'$ to $M$, and thus  $R_{s',M}\le \eta + R_{\sigma,M}$ as claimed.

So if we choose $\eta\ge R_{\sigma, M'}$, then the probability of phase estimation successfully outputting $0^t$ as in Lemma~\ref{lem:PEG'} is at least 
\[\frac{\eta}{R_{s',M}} \ge \frac{\eta}{\eta+R_{\sigma,M}}\ge\frac{1}{2}\]

Algorithm~\ref{alg:findeta} therefore terminates by the time $\eta= R_{\sigma,M}$.
The total runtime of Algorithm~\ref{alg:findeta} is:
\[\sum_{i=1}^{\log(R_{\sigma,M}W)} \sqrt{2^i} = O(\sqrt{R_{\sigma,M}W})\]

\subsubsection{Estimating the effective resistance $R_{s,M}$}
\label{sec:estimateR}
Note that this algorithm provides an estimate of the effective resistance in the original graph $G$ when the initial distribution $\sigma$ is concentrated at a single vertex $s$, because in this case $R_{s',M}=\eta+R_{s,M}$.
Given a constant multiplicative accuracy approximation $\tilde{a}$ of $\eta/R_{s',M}$, and exact knowledge of $\eta$, we can therefore estimate $R_{s,M}$ to constant multiplicative accuracy.

To obtain a more accurate estimate of $R_{s,M}$ to multiplicative precision $\epsilon$, we could do one more iteration of Algorithm~\ref{alg:findeta} with $\eta$ set to the constant accuracy approximation for $R_{s,M}$ that has been obtained so far. Then running phase estimation for $O(\sqrt{R_{s,M}W}/\epsilon)$ time will leave the ancilla register in the state $\ket{0^t}$ with probability $\eta/(\eta+R_{s,M})$ up to accuracy $\epsilon$ by Lemma~\ref{lem:PEG'}. Finally, performing amplitude estimation to accuracy $\epsilon$ requires $O(1/\epsilon)$ iterations, giving a total runtime of $O(\sqrt{R_{s,M}W}/\epsilon^2)$, as claimed in Theorem~\ref{thm:estimateR}.

We believe that this runtime could be improved to $O(\sqrt{R_{s,M}W}\frac{1}{\epsilon}\log(\frac{1}{\epsilon}))$ by performing phase estimation for time $O(\sqrt{R_{s,M}W})$ repeatedly  $O(\log(\frac{1}{\epsilon}))$ times (and still performing amplitude estimation to accuracy $\epsilon$). This is because if all $O(\log(\frac{1}{\epsilon}))$ instances of phase estimation return the zero string, then we would we expect to have prepared the state $\ket{\Phi}$ to accuracy ${\epsilon}$, but we do not provide a full analysis of this claim here.

\subsection{Simple algorithm with $O(\poly(|M|)$ scaling}
\label{sec:simplealg}
Now that we have found an appropriate choice of $\eta$ , we would like to choose a value of $x$ that satisfies the following two conditions:
\begin{enumerate}
\item $R_{s',M'}=O(R_{\sigma,M})$ so that we can (approximately) prepare $\ket{\Phi}$ in time $O(\sqrt{R_{\sigma,M}W})$.
\item $\dfrac{x\sum_{k \in M} f_{kk'}^2}{R_{s',M'}}=\Omega(1)$ so that sampling $\ket{\Phi}$ returns a marked vertex with good probability.
\end{enumerate}
Unfortunately we are not in general able to find any such choice of $x$ that satisfies both of these conditions. In the rest of this section we describe a natural, but suboptimal choice of $x$, and in Section~\ref{sec:fullalg} we describe how to improve on this choice.

A natural choice for $x$ is the value $\tilde{\eta}$ returned by Algorithm~\ref{alg:findeta}. 
We first prove that this choice of $x$ satisfies Condition 1: 
note that $\tilde{\eta}=\Theta(R_{s'M})$ and $\tilde{\eta} \le R_{\sigma,M}$ as argued in Section~\ref{sec:findeta}; so we have that $R_{s',M} =O(R_{\sigma,M})$. It remains to show that $R_{s',M'} =O(R_{s',M})$. To show this, observe that the possible extra energy between $M$ and $M'$ in the definition of $R$ (Eq(\ref{eq:energy})) is 
\[x\sum_{k \in M} f_{kk'}^2\le x=\tilde{\eta} \le R_{\sigma,M}\]
where the first inequality holds because $f$ is a unit flow.
Thus $R_{s',M'}=O(R_{\sigma,M})$ as claimed.

Considering Condition 2, we now have 
\[\dfrac{x\sum_{k \in M} f_{kk'}^2}{R_{s',M'}}=\Omega(1) \sum_{k \in M} f_{kk'}^2\]
but the problem is that the best lower bound we can put on $\sum_{k \in M} f_{kk'}^2$ is $\Omega(1/|M|)$, because the flow may be evenly distributed amongst the vertices in $M$.

To find a marked vertex would therefore require preparing $\ket{\Phi}$ to accuracy $1/|M|$ (which by Lemma~\ref{lem:PEG'} takes time $O(\sqrt{R_{\sigma,M}W}|M|^2)$) and repeatedly sampling from this state $O(|M|)$ times until a marked element is found. This algorithm is described in Algorithm~\ref{alg:simple} and takes a total of $O(\sqrt{R_{\sigma,M}W}|M|^3)$ time.
%

\begin{algorithm}
\caption{Simple method to find a marked vertex }
\label{alg:simple}
\KwIn{Graph $G$, initial distribution $\sigma$ over vertices, upper bound $W$ on total weight, upper bound $m$ on $|M|$}
Run Algorithm~\ref{alg:findeta} to find $\tilde{\eta}$ such that $\tilde{\eta}/R_{s',M}\approx 1/2$\;
Then set $x=\eta=\tilde{\eta}$, and construct $G'$ and $U_AU_B$ for this choice of $\eta,x$\;
Repeat the following until success: \\
Run phase estimation for time $O(\sqrt{\tilde{\eta}W}m^2)$. Then measure and check if element is marked.
\end{algorithm}

\subsection{Full algorithm with improved dependence on $|M|$}
\label{sec:fullalg}
In order to come as close as possible to satisfying Conditions 1 and 2 from Section~\ref{sec:simplealg}, we want to understand how $x$ and $R_{s',M'}$ and $\sum_{k \in M} f_{kk'}^2$ are related. Fortunately we are able to derive the following simple relation:
%
%

\begin{lem}
\label{lem:Rxq}
	Let $R_{s',M'}(x)$ be the resistance of the electrical flow when additional edges of resistance $x$ are added to each marked vertex. Let $q(x)=\sum_{k \in M} |f_{kk'}|^2$ be the $l_2$ norm of the flow down these additional edges. 
	Then 
	\begin{enumerate}
		\item $R_{s',M'}(x)$ is convex: $R_{s',M'}(x+y)\le R_{s',M'}(x)+yq(x)$ for all $y$.
		\item $\dfrac{d}{dx}R_{s',M'}=q(x)$
	\end{enumerate}
\end{lem}

\begin{proof}
	Consider the electric flow from $s'$ to $M'$ in the graph $G'(x)$. It is easy to check that the energy of this flow in the graph $G'(x+y)$ is exactly  $R_{s',M'}(x)+yq(x)$.
	But the minimal energy of all such flows from $s'$ to $M'$ in $G'(x+y)$ is exactly the effective resistance $R_{s',M'}(x+y)$. Therefore we have proven 1.:
	\[ R_{s',M'}(x+y)\le R_{s',M'}(x)+yq(x)\quad \forall y\]
Now we prove 2. By a relabelling $x+y=z$ and $x=z+w$, we have 
$R_{s',M'}(z)\le R_{s',M'}(z+w)-wq(z+w)$, and so:
\[ q(x+y) \le \frac{R_{s',M'}(x+y)-R_{s',M'}(x)}{y} \le q(x)\]
Taking the limit $y \rightarrow 0$ completes the proof, noting that $q(x)$ is a continuous function of the electric flow which is determined by a set of linear equations and hence $q(x)$ is itself continuous.
\end{proof}

The idea of the algorithm is then to sample $x$ from an interval $[a,b]$ with probability $1/(x\log(b/a))$. Then the expected probability of hitting a marked vertex when sampling from $\ket{\Phi}$ is 
\[\int_a^b \frac{x q(x)}{R_{s',M'}(x)} \frac{1}{x\log(b/a)} dx
\ge\frac{1}{\log(b/a)R_{s',M'}(b)}\int_a^b q(x)dx=\frac{R_{s',M'}(b)-R_{s',M'}(a)}{\log(b/a)R_{s',M'}(b)}\]
where the first inequality holds because $R_{s',M'}(x)$ increases monotonically with $x$, and the equality follows from Lemma~\ref{lem:Rxq}.

If the interval $[a,b]$ is such that $R_{s',M'}(x)$ increases by a constant multiplicative factor across this interval (but still satisfying Condition 1: $R_{s',M'}(b)=O(R_{\sigma,M})$), then the expected probability of finding a marked vertex is $\Omega(1/\log(b/a))$. It therefore suffices to prepare $\ket{\Phi}$ to accuracy $1/\log(b/a)$, which can be done in time $O(\sqrt{R_{\sigma,M}W}\log^2(b/a))$ by Lemma~\ref{lem:PEG'}. 
Then measure until a marked vertex is found, which will take $O(\log(b/a)$ tries, resulting in a total runtime of $O(\sqrt{R_{\sigma,M}W}\log^3(b/a))$.

In Algorithm~\ref{alg:findvertex}, we first find such an interval $[a,b]$, by setting $a=\tilde{\eta}=\Theta(R_{\sigma,M})$ the output of Algorithm~\ref{alg:findeta}. 
Then the algorithm doubles $x$ until it finds that the output of amplitude estimation (which approximates $\eta/R_{s',M'}$) has  halved, implying that $R_{s',M'}$ has doubled.
Each step of this part of the algorithm takes time $O(\sqrt{R_{\sigma,M}W})$ and there are $\log(b/a)$ steps, so this part of the algorithm only takes $O(\sqrt{R_{\sigma,M}W}\log(b/a))$ time.
\begin{algorithm}
\caption{Find a marked vertex with improved dependence on $|M|$}
\label{alg:findvertex}
\KwIn{Graph $G$, initial distribution $\sigma$ over vertices, upper bound $W$ on total weight}
Run Algorithm~\ref{alg:findeta} to find $\tilde{\eta}$ such that  $\tilde{\eta}/R_{s',M}\approx 1/2$ \\
Then set $x=\eta=\tilde{\eta}$\\

Run phase estimation for time $O(\sqrt{\tilde{\eta}W})$ and amplitude estimation on $\ket{0^t}$ to obtain an estimate for $\tilde{\eta}/R_{s',M'}$\\
Double $x$ and return to step 3. until outcome of amplitude estimation has halved.\\
Set $a=\tilde{\eta}$ and $b$ equal to the final value of $x$ in Step 4\\
Now repeat until success: Fix $\eta=\tilde{\eta}$ and sample $x \in [a,b]$ with pdf $1/x\log(b/a)$.
Construct $U_AU_B$ corresponding to this choice of $x, \eta$. Run phase estimation for time $O(\sqrt{\tilde{\eta}W}\log^2(b/a))$. Then measure and check if element is marked.
\end{algorithm}

So Algorithm~\ref{alg:findvertex} finds a marked vertex in time $O(\sqrt{R_{\sigma,M} W} \log^3(b/a))$. It only remains to show that $b/a=O(|M|)$ to prove Theorem~\ref{thm:findm}.
 
Recall that $q(x)\ge 1/|M|$ because in the worst case the flow is evenly distributed between the $|M|$ marked vertices. By Lemma~\ref{lem:Rxq} 
\[\frac{dR_{s',M'}}{dx}=q(x)\ge\frac{1}{|M|} \]
and therefore
\[ R_{s',M'}(b)-R_{s',M'}(a)\ge\frac{b-a}{|M|}\quad \Rightarrow \quad \frac{b}{a} \le 1+|M|\frac{R_{s',M'}(b)-R_{s',M'}(a)}{a}\]
It remains to observe that $R_{s',M'}(b)\approx 2R_{s',M'}(a)$ and that for $a=\tilde{\eta}$
\[\frac{R_{s',M'}(a)}{a}=O\left(\frac{R_{\sigma,M}}{\tilde{\eta}}\right)=O(1)\]

\section*{Acknowledgements}
I would like to thank Ashley Montanaro, Joran Van Apeldoorn and Simon Apers for some very interesting discussions which contributed to this project.

\bibliographystyle{abbrv}
\bibliography{quantumwalks}

\begin{thebibliography}{10}

\bibitem{Ambainis2019}
A.~Ambainis, A.~Gily{\'e}n, S.~Jeffery, and M.~Kokainis.
\newblock Quadratic speedup for finding marked vertices by quantum walks, 2019.
\newblock \texttt{arXiv:1903.07493}.

\bibitem{Belovs2013}
A.~Belovs.
\newblock Quantum walks and electric networks.
\newblock \texttt{\href{https://arxiv.org/abs/1302.3143}{arXiv:1302.3143}},
  Feb. 2013.

\bibitem{Bollobas1998}
B.~Bollobas.
\newblock {\em Modern Graph Theory}, volume 184.
\newblock Springer Science \& Business Media, 1998.

\bibitem{Cleve1998}
R.~Cleve, A.~Ekert, C.~Macchiavello, and M.~Mosca.
\newblock Quantum algorithms revisited.
\newblock {\em Proceedings of the Royal Society of London. Series A:
  Mathematical, Physical and Engineering Sciences}, 454(1969):339--354, 1998.

\bibitem{Dohotaru2017}
C.~Dohotaru and P.~H{\o}yer.
\newblock Controlled quantum amplification.
\newblock In {\em 44th International Colloquium on Automata, Languages, and
  Programming (ICALP 2017)}. Schloss Dagstuhl-Leibniz-Zentrum fuer Informatik,
  2017.

\bibitem{Gilyen2019}
A.~Gily{\'e}n, S.~Jeffery, and S.~Apers.
\newblock A unified framework of quantum walk search, 2019.
\newblock Personal communication.

\bibitem{Ito2019}
T.~Ito and S.~Jeffery.
\newblock Approximate span programs.
\newblock {\em Algorithmica}, 81(6):2158--2195, 2019.

\bibitem{Jarret2017}
M.~Jarret and K.~Wan.
\newblock Improved quantum backtracking algorithms using effective resistance
  estimates.
\newblock {\em Physical Review A}, 97(2), nov 2018.
\newblock \texttt{\href{https://arxiv.org/abs/1711.05295}{arXiv:1711.05295}}.

\bibitem{Kitaev1995}
A.~Y. Kitaev.
\newblock Quantum measurements and the abelian stabilizer problem, 1995.
\newblock \texttt{arXiv:quant-ph/9511026}.

\bibitem{Krovi2016}
H.~Krovi, F.~Magniez, M.~Ozols, and J.~Roland.
\newblock Quantum walks can find a marked element on any graph.
\newblock {\em Algorithmica}, 74(2):851--907, 2016.

\bibitem{Lee2010}
T.~Lee, R.~Mittal, B.~W. Reichardt, R.~{\v{S}}palek, and M.~Szegedy.
\newblock Quantum query complexity of state conversion.
\newblock In {\em Proceedings - Annual IEEE Symposium on Foundations of
  Computer Science, FOCS}, pages 344--353, nov 2011.

\bibitem{Montanaro2015}
A.~Montanaro.
\newblock Quantum walk speedup of backtracking algorithms.
\newblock sep 2015.
\newblock \texttt{\href{https://arxiv.org/abs/1509.02374}{arXiv:1509.02374}}.

\bibitem{Szegedy2004}
M.~Szegedy.
\newblock Quantum speed-up of markov chain based algorithms.
\newblock In {\em 45th Annual IEEE symposium on foundations of computer
  science}, pages 32--41. IEEE, 2004.

\end{thebibliography}

\appendix
\section{Phase estimation analysis}
\label{sec:PE}
In this section we analyse how phase estimation performs in the case where we have a bound of the form provided by the effective spectral gap lemma.

\begin{lem}
\label{lem:PEnew}
Let $U$ be a unitary and let $P_{\epsilon}$ be the projector onto eigenstates of $U$ with eignvalue $e^{2i \theta}$ for $|\theta|  \le \epsilon$.
Let $\ket{\phi}$ be an eigenvector of $U$ with eigenvalue $1$, and $\ket{\psi}$ be a normalised state such that $\norm{P_{\epsilon}\ket{\psi}-\sqrt{p}\ket{\phi}}\le \epsilon C$ for all $\epsilon$. 
Then performing phase estimation with $t$ ancilla bits takes time $O(2^t)$ and outputs $0^t$ with probability $p' \in [p,p+ O\left(\frac{C}{2^t}\right)]$ leaving a state $\ket{\phi'}$ such that 
\[ \frac{1}{2}\norm{\proj{\phi'}-\proj{\phi}}_1=O\left( \sqrt{\frac{C}{p'2^t}}\right)\]
\end{lem}
\begin{proof}
We do phase estimation to $t$-bits of accuracy, which takes time $2^t$ (see Lemma~\ref{lem:PE}). 
We decompose $\ket{\psi}$ in the eigenbasis of $U$:
\[ \ket{\psi}= p\ket{\phi}+\sum_{k}a_k \ket{\psi_k} \qquad \text{ where } U\ket{\psi_k}= e^{2i \theta_k}\ket{\psi_k}\]
We now consider the condition $\norm{P_{\epsilon}\ket{\psi}-\sqrt{p}\ket{\phi}}\le \epsilon C$. Taking $\epsilon=0$, we see that $P_0\ket{\psi}=\sqrt{p}\ket{\phi}$ and hence that $\theta_k \in (0,\pi)$ for all $k$.
For $\epsilon \ge 0$ we can write this condition as:
\begin{equation}
\label{eq:smalltheta}
\sum_{k: 0<|\theta_k| \le \epsilon} a_k^2 \le \epsilon^2 C^2.
\end{equation}
The other inequality that we will need comes from the fact that for large values of $\theta$, phase estimation is unlikely to return a zero answer. 
By part 3 of Lemma~\ref{lem:PE}, $V\ket{\psi_k}\ket{0^t}=\ket{\psi_k}\ket{w_k}$ for some $\ket{w_k}$. Let $\mu_k=|\braket{w_k|0^t}|$. Then there exists a universal constant $\gamma$ such that 
\begin{equation}
\label{eq:largetheta}
\mu_k^2 \le \frac{\gamma}{2^t \epsilon} \quad \forall k: |\theta_k| \ge \epsilon
\end{equation}
Then the probability of getting $\ket{0^t}$ after measuring the ancilla register is  
\[p'=p+\sum_k |a_k \mu_k|^2=p+\sum_{k: |\theta_k| \le 1/C} a_k^2\mu_k^2 +\sum_{k: |\theta_k| > 1/C} a_k^2\mu_k^2 \]

We bound the first part of this sum using Eq(\ref{eq:largetheta}), and use the fact that $\ket{\psi}$ is normalised:
\[\sum_{k: |\theta_k| > 1/C} a_k^2\mu_k^2 \le \sum_{k: |\theta_k| > 1/C} a_k^2 \frac{\gamma C}{2^t}\le \sum_{k} a_k^2 \frac{\gamma C}{2^t}\le\frac{\gamma C}{2^t}\]
To bound the second part of this sum we first further subdivide the sum into regions of ever decreasing $\theta$, and use Eq(\ref{eq:largetheta}) and then Eq(\ref{eq:smalltheta}).
\[\sum_{k: |\theta_k| \le 1/C} a_k^2\mu_k^2= \sum_{m=0}^{\infty} \left(\sum_{k: \frac{1}{C2^{m+1}} \le |\theta_k| \le \frac{1}{C2^m}}a_k^2\mu_k^2 \right)
\le  \sum_{m=0}^{\infty} \left(\sum_{k: \frac{1}{C2^{m+1}} \le |\theta_k| \le \frac{1}{C2^m}}a_k^2 \frac{\gamma C2^{m+1}}{2^t}\right)\]
\[\le \sum_{m=0}^{\infty} \left(\sum_{k: |\theta_k| \le \frac{1}{C2^m}}a_k^2 \frac{\gamma C2^{m+1}}{2^t}\right)
 \le \sum_{m=0}^{\infty}  \frac{1}{2^{2m}}\frac{\gamma C2^{m+1}}{2^t}=\frac{2\gamma C}{2^t}\]
 The final equality is a standard geometric series calculation.
 
 Finally, if the $0^t$ outcome is observed, then the post measurement state on the remaining register is of the form
 \[\ket{\phi'}=\frac{1}{\sqrt{p'}}\left(\sqrt{p}\ket{\phi}+ \sum_k a_k^2 \mu_k^2\ket{\phi^{\perp}}\right) \]
for some state $\ket{\phi^{\perp}}$ orthogonal to $\ket{\phi}$.
Then 
\[ \frac{1}{2}\norm{\proj{\phi'}-\proj{\phi}}_1=\sqrt{1-|\braket{\phi'|\phi}|^2}=\sqrt{1-\frac{p}{p'}} =\sqrt{\frac{\sum_k a_k^2 \mu_k^2}{p'}} =O\left( \sqrt{\frac{C}{p'2^t}}\right)\]

\end{proof}
\section{No knowledge of $W$}
\label{sec:unknownW}
The algorithms described above all rely on knowledge of an upper bound of $W$ in order to run.
If no upper bound on $W$ is known, then we can simply guess different values of $W$ with only a small difference to the run time.

\newcommand{\Wmin}{W_{\text{min}}}
\newcommand{\Rmin}{R_{\text{min}}}

Let $\Wmin$ and $\Rmin$ be lower bounds for $W$ and $R_{\sigma,M}$ respectively.
For a particular value of $T$, try all choices for  $W$ in the set$\{ \Wmin, 2\Wmin, \dots , T^2/\Rmin \}$ and run for a time $T$. If none of these choices for $W$ work, double $T$ and start again.

For a fixed $T$ each stage takes a total time $T\log\left(\frac{T^2}{\Rmin \Wmin}\right)$. With high probability the algorithm completes when $T= O(\sqrt{RW}\log |M| )$ so the total time taken is therefore:
\[\sum_{T=1}^{O(\sqrt{RW}\log |M| )}T\log\left(\frac{T^2}{\Rmin \Wmin}\right)=O\left(\sqrt{RW}\log |M| \log \left(\frac{\sqrt{RW}\log|M|}{\Rmin \Wmin} \right)\right)\]
For the special case where $\sigma$ is concentrated entirely at a single vertex $s$, a sensible choice would be $\Rmin=1/d_s$ and $\Wmin=d_s$ where $d_s$ is the weighted degree of $s$ in $G$. 
These are indeed lower bounds for $R_{s,M}$ and $W$, and furthermore this choice implies that $\Rmin \Wmin=1$, resulting in a overall runtime of $\tilde{O}(\sqrt{RW}\log |M|)$.

This runtime can also be achieved for more general distributions $\sigma$, by setting $\Rmin=1/d_{\sigma}$ and $\Wmin=d_{\sigma}$ where we define $d_{\sigma}$, (which we may think of as the weighted degree of $\sigma$),  by 
\begin{equation}
\frac{1}{d_\sigma}=\sum_{u \in V} \frac{\sigma_u^2}{d_u}
\label{eq:dsigma}
\end{equation}
This quantity $d_{\sigma}$ may be difficult to compute, but we can show that it provides the necessary lower bounds for $R$ and $W$.

\begin{lem}
For $d_{\sigma}$ defined in Eq (\ref{eq:dsigma}):
\[R \ge \frac{1}{d_{\sigma}} \ge \frac{1}{W}\] 
and the second inequality is an equality if and only if $\sigma$ is the stationary distribution on $G$.
\end{lem}

\begin{proof}
For a given vertex $u$, consider the edges $uv$ for which there is flow from $u$ to $v$ in the electric flow. Then the minimum energy of a unit flow out of $u$ along these edges is $1/d_u$ 
For each vertex $u$, there is $\sigma_u$ flowing out of $u$ in the electrical flow. Considering the contribution to the energy (defined in Eq (\ref{eq:energy})) of just the edges attached to $u$,  which even if evenly distributed contributes an energy of $\sigma_u^2/d_u$ to the resistance.

For the second inequality, observe that 
\begin{align}\frac{1}{d_{\sigma}}= \sum_{u \in V} \frac{\sigma_u^2}{d_u}
&=\sum_{u \in V} \left[\frac{\left(\sigma_u-d_u/W\right)^2}{d_u} +\frac{2d_u \sigma_u}{Wd_u} -\frac{d_u^2}{d_u W^2}\right]\\
&\ge 0+\frac{2}{W}\sum_{u \in V} \sigma_u -\frac{1}{W}\sum_{u \in V} \frac{d_u}{W}= \frac{1}{W}
\end{align}
where we have used the inequality $(\sigma_u-d_u/W)^2 \ge 0$ and the fact that $\sigma_u$ and $d_u/W$ are probability distributions which sum to $1$.
This inequality is tight if and only if $\sigma_u = d_u/W$ for all $u$, which is the stationary distribution.
\end{proof}

	\end{document}